\documentclass[12pt]{iopart}

\usepackage{graphicx}
\usepackage{bm}
\usepackage{times}
\usepackage{amsmath,amssymb}
\usepackage{amsthm}
\usepackage{enumerate}
\usepackage{multirow}

\usepackage[numbers,sort&compress]{natbib}

\usepackage[colorlinks=true,linkcolor=blue,citecolor=red, linktocpage=true,breaklinks=true]{hyperref}
\usepackage{cleveref}

\newtheorem{theorem}{Theorem}[section]
\newtheorem{lemma}[theorem]{Lemma}
\newtheorem{proposition}[theorem]{Proposition}

\newtheorem{definition}[theorem]{Definition}

\numberwithin{equation}{section}

\newcommand{\ket}[1]{|#1\rangle}
\newcommand{\bra}[1]{\langle #1|}

\begin{document}

\title{Asymptotic optimality of Grover--Radhakrishnan--Korepin algorithm}

\author{Kun Zhang$^{1, 2, 3, 4, \dag}$, Kang-Yuan Chen$^{1}$, Xiao-Hui Wang$^{1, 2, 3, 4}$, Vladimir Korepin$^{5}$}

\address{$^1$ School of Physics, Northwest University, Xi’an 710127, China}
\address{$^2$ Shaanxi Key Laboratory for Theoretical Physics Frontiers, Xi'an 710127, China}
\address{$^3$ Peng Huanwu Center for Fundamental Theory, Xi'an 710127, China}
\address{$^4$ Fundamental Discipline Research Center for Quantum Science and technology of Shaanxi Province, Xi'an 710127, China}
\address{$^5$ C.N. Yang Institute for Theoretical Physics, Stony Brook University, Stony Brook, New York 11794, USA}

\ead{kunzhang@nwu.edu.cn}


\begin{abstract}
Grover's algorithm is a cornerstone of quantum algorithms and is strictly optimal in oracle-query complexity. While the full search problem admits no further improvement, one may trade accuracy for speed in the partial search problem, where the task is to identify only the block containing the target item. The best known quantum algorithm for the partial search problem is the Grover--Radhakrishnan--Korepin (GRK) algorithm, whose optimality has long been conjectured but not proved. In this work, we prove the optimality of GRK in the large-block limit. We formulate partial search as a time-optimal control problem and apply the Pontryagin maximum principle to derive the switching-function dynamics, establish the bang--bang structure of regular extremals, and exclude non-optimal switching patterns. As a result, we show that the optimal regular extremal has the global--local--global form, which yields a control-theoretic proof of the asymptotic optimality of the GRK algorithm in oracle-query complexity.
\end{abstract}
%
%

\section{\label{sec:intro}Introduction}

Grover's algorithm solves the unstructured search problem with \(\mathcal O(\sqrt{N})\) oracle queries, yielding a quadratic speedup over classical search \cite{grover1996fast,Grover1997needle}. This query complexity is strictly optimal \cite{boyer1996tight,zalka1999grover}. Beyond its foundational role in quantum algorithms, Grover search has important applications such as the key search \cite{Grassl2016ApplyingGroverAES,Jaques2020GroverAESLowMC,Jang2025QuantumAnalysisAES}. Accordingly, cryptographic schemes whose security is reduced to exhaustive key search are typically expected to lose at most a quadratic security factor against quantum attacks. Grover's algorithm has also been experimentally demonstrated on a variety of quantum computing platforms \cite{Chuang1998,zhang2021implementation,pokharel2024better}.

Since Grover's algorithm is optimal in oracle-query complexity, any genuine improvement over full Grover search must come from relaxing the task itself rather than from improving the standard unstructured search problem. A natural relaxation is the partial search problem, where the goal is to determine only part of the target address, or equivalently the target block in a partitioned database. Grover and Radhakrishnan showed that this task can be solved faster than full quantum search \cite{GroverRadhakrishnan2005PartialSearch}. Their algorithm was later simplified and optimized, leading to the Grover--Radhakrishnan--Korepin (GRK) algorithm \cite{Korepin2005OptimizationPartialSearch,KorepinGrover2006SimplePartialSearch}. The GRK algorithm employs two types of Grover operators, usually called global and local operators, and admits a three-dimensional formulation, in contrast to the two-dimensional formulation of the standard Grover search \cite{KorepinVallilo2006GroupTheoreticalPartialSearch}. Multi-target extensions of partial quantum search were developed in \cite{Choi2007QuantumPS,zhong2009quantum,Zhang2017QuantumPS}. Exact versions of the GRK algorithm with unit success probability were also constructed in \cite{ChoiWalkerBraunstein2007SureSuccessPartialSearch,YeWang2025DeterministicPartialSearchOneSixteenth}.

Since the introduction of the GRK algorithm, the reason for the optimal global--local--global structure has remained unclear. For large databases and large blocks, the number of admissible operator sequences grows exponentially, and whether GRK is optimal among them has remained open. Existing analytical results cover only a few restricted classes of sequences and support the optimal GRK structure \cite{KorepinLiao2006QuestFastPartialSearch}. Recent numerical work provides further strong evidence, but still no proof \cite{JiangWangZhangKorepin2026ExactBoundsPartialSearch}.

In the large-block limit, the search for a partial-search protocol with minimal oracle complexity can be recast as a time-optimal control problem. In this regime, the global and local Grover stages reduce to two fixed generators on a three-dimensional state space, and the oracle-query cost becomes the total control time. The central question is therefore to determine the switching structure of optimal trajectories for this two-generator system. The relevant tool is the Pontryagin maximum principle (PMP), which gives first-order necessary conditions for optimality through a maximized PMP function and switching functions \cite{Pontryagin1962,Liberzon2012,AgrachevSachkov2004}. PMP and related geometric-control methods have been widely used in quantum control, especially in time-optimal state transfer and gate synthesis \cite{DongPetersen2010,KhanejaGlaserBrockett2001,Glaser2015TrainingSchrodinger,Stefanatos2020RolandCerf,BoscainMasonSigalotti2021}. 

In this paper, we resolve the optimality question for the GRK algorithm in the asymptotic regime (large database and large block). We formulate partial search as a time-optimal switching problem, derive the reduced PMP dynamics, and analyze the resulting bang--bang extremals. By excluding non-optimal switching patterns and reducing admissible extremals, we show that the optimal regular structure is global--local--global. This yields a control-theoretic proof of the asymptotic optimality of the GRK algorithm in oracle-query complexity.

The paper is organized as follows. Section \ref{sec:GRK} reviews the GRK algorithm and the three-dimensional reduction. Section \ref{sec:PMP} formulates the time-optimal control problem and introduces the PMP. Section \ref{sec:switching_dynamics} derives the reduced switching dynamics. Section \ref{sec:singluar_arcs} analyzes switching compressions and singular arcs. Section \ref{sec:optimality} presents the final proof of the optimality of the GRK algorithm. Section \ref{sec:conclusion} concludes with a discussion. \ref{app:endpoint} proves the endpoint bang structure used in Sec. \ref{sec:optimality}.

\section{\label{sec:GRK}GRK algorithm and the three-dimensional formulation}

In Sec. \ref{subsec:GRK}, we recall the GRK partial-search algorithm and fix the notation used throughout the paper. In Sec. \ref{subsec:three}, we then describe its standard three-dimensional formulation and the associated infinitesimal generators that will underlie the control-theoretic analysis.

\subsection{\label{subsec:GRK}The GRK partial-search algorithm}

We begin by recalling the standard formulation of quantum partial search and the GRK algorithm. Consider an unstructured database of size $N=2^n$ with a unique marked item $\ket{t}$. In the partial-search problem one does not seek the full $n$-bit target string, but only a partial specification of it. Accordingly, we decompose the target string as $t=t_1t_2$, where $t_1$ has length $n-m$ and $t_2$ has length $m$. We denote the number of blocks and the size of each block by
\begin{equation}\label{eq:block-parameters-sec2}
K=2^{\,n-m},
\qquad
b=2^m=\frac{N}{K}.
\end{equation}
Thus the database is partitioned into $K$ blocks of equal size $b$, and the goal is to identify the block labeled by $t_1$.

Quantum search algorithms use an oracle to identify the target state. Using phase kickback, the oracle can be implemented as $O_t\ket{x}=(-1)^{f(x)}\ket{x}$, where $f(t)=1$ and $f(x)=0$ for $x\neq t$, equivalently
\begin{equation}
O_t=I-2\ket{t}\!\bra{t}.
\end{equation}
The initial state is the uniform superposition
\begin{equation}
\ket{s_n}=\frac{1}{\sqrt N}\sum_{x=0}^{N-1}\ket{x}.
\end{equation}

Another key element of quantum search is the diffusion operator, defined as
\begin{equation}
D_n=2\ket{s_n}\!\bra{s_n}-I,
\end{equation}
and the corresponding Grover operator is
\begin{equation}
G_n=D_nO_t.
\end{equation}
Iterative application of the Grover operator to the initial state $\ket{s_n}$ amplifies the amplitude of the target state, yielding Grover's algorithm \cite{grover1996fast,Grover1997needle}. This requires only $\mathcal O(\sqrt N)$ oracle calls and hence provides a quadratic speedup over classical search.

For partial search one also introduces the local diffusion operator \cite{GroverRadhakrishnan2005PartialSearch}
\begin{equation}
D_m=I_{n-m}\otimes \bigl(2\ket{s_m}\!\bra{s_m}-I_m\bigr),
\end{equation}
where $\ket{s_m}$ is the uniform superposition on a single block. The associated local Grover operator is
\begin{equation}
G_m=D_mO_t.
\end{equation}
Thus $G_n$ acts globally on the full database, whereas $G_m$ performs Grover rotations independently inside each block.

The GRK algorithm has the three-step form $G_n\,G_m^{k_2}\,G_n^{k_1}\ket{s_n}$, namely: first $k_1$ global Grover iterations, then $k_2$ local Grover iterations, and finally one additional global Grover iteration \cite{GroverRadhakrishnan2005PartialSearch}. The success condition for partial search is that the final state has no amplitude on any non-target block. Using the block parameters defined in Eq.~(\ref*{eq:block-parameters-sec2}), this is achieved in the asymptotic large-block regime by choosing
\begin{equation}\label{eq:k1k2-asymp-sec2}
k_1=\frac{\pi}{4}\sqrt{N}-\eta\sqrt{b},
\qquad
k_2=\alpha\sqrt{b},
\end{equation}
with parameters $\eta,\alpha$ constrained by the success condition
\begin{equation}\label{eq:constraint-alpha-eta-sec2}
\tan\!\left(\frac{2\eta}{\sqrt K}\right)
=
\frac{2\sqrt K\,\sin(2\alpha)}{K-4\sin^2\alpha}.
\end{equation}
Optimizing the total number of oracle queries inside the global--local--global family yields the well-known GRK parameters \cite{Korepin2005OptimizationPartialSearch,KorepinGrover2006SimplePartialSearch}
\begin{equation}\label{eq:GRK-optimal-parameters-sec2}
\tan\!\left(\frac{2\eta_K}{\sqrt K}\right)
=
\frac{\sqrt{3K-4}}{K-2},
\qquad
\cos(2\alpha_K)=\frac{K-2}{2(K-1)}.
\end{equation}
Accordingly, the optimized GRK algorithm uses
\begin{equation}
k_1+k_2+1
=
\frac{\pi}{4}\sqrt N-(\eta_K-\alpha_K)\sqrt b+\mathcal O(1)
\end{equation}
queries. In the large-$K$ limit one has $\alpha_K\to\pi/6$ and $\eta_K\to \sqrt 3/2$, hence 
\begin{equation}
k_1=\frac{\pi}{4}\sqrt N-\frac{\sqrt{3b}}{2},
\qquad
k_2=\frac{\pi}{6}\sqrt b.
\end{equation}

Other operator orderings, such as global--local and local--global--local--global patterns, have also been investigated \cite{KorepinLiao2006QuestFastPartialSearch}. More recent exhaustive numerical studies provide strong evidence that, among operator sequences with a fixed total number of oracle queries, the GRK ordering is optimal and is attained within the global--local--global family \cite{JiangWangZhangKorepin2026ExactBoundsPartialSearch}. The issue left open is therefore structural rather than parametric, namely why the global--local--global pattern is the optimal one for quantum partial search.

\subsection{\label{subsec:three}Three-dimensional formulation of the partial-search problem}

The partial-search problem admits a natural three-dimensional formulation \cite{KorepinVallilo2006GroupTheoreticalPartialSearch}. We introduce the orthonormal basis
\begin{align}
&\ket{t}=\ket{t_1}\otimes \ket{t_2}, \label{eq:basis-t-sec2}\\
&\ket{ntt}
=
\frac{1}{\sqrt{b-1}}
\sum_{j\neq t_2}\ket{t_1}\otimes \ket{j}, \label{eq:basis-btbar-sec2}\\
&\ket{u}
=
\frac{1}{\sqrt{N-b}}
\Bigl(
\sqrt N\,\ket{s_n}-\ket{t}-\sqrt{b-1}\,\ket{ntt}
\Bigr). \label{eq:basis-bbar-sec2}
\end{align}
Here $\ket{ntt}$ is the normalized equal superposition of all non-target states inside the target block, while $\ket{u}$ is the normalized equal superposition of all states in the non-target blocks.

The relevant invariant subspace is therefore $\mathcal H_{\mathrm{red}}=\mathrm{span}\{\ket{t},\ket{ntt},\ket{u}\}$. With respect to the ordered basis $\{\ket{t},\ket{ntt},\ket{u}\}$, the initial state $\ket{s_n}$ takes the form
\begin{equation}\label{eq:sn-reduced-sec2}
\ket{s_n}
=
\sin\gamma\,\sin\theta_2\,\ket{t}
+\sin\gamma\,\cos\theta_2\,\ket{ntt}
+\cos\gamma\,\ket{u},
\end{equation}
where
\begin{equation}\label{eq:angles-sec2}
\sin\theta_1=\frac{1}{\sqrt N},
\qquad
\sin\theta_2=\frac{1}{\sqrt b},
\qquad
\sin\gamma=\frac{1}{\sqrt K}.
\end{equation}
The terminal condition for partial search is simply that the amplitude on $\ket{u}$ vanish. Thus the terminal set is the plane
\begin{equation}\label{eq:Sigma-sec2}
\Sigma
=
\left\{
\ket{\psi}\in\mathcal H_{\mathrm{red}}:\ \langle u|\psi\rangle=0
\right\}.
\end{equation}

With respect to the ordered orthonormal basis $\{\ket{t},\ket{ntt},\ket{u}\}$, the global and local Grover operators are represented by the following $3\times 3$ matrices \cite{KorepinVallilo2006GroupTheoreticalPartialSearch}:
\begin{equation}\label{eq:Gn-reduced-sec2}
G_n=
\begin{pmatrix}
1-2\sin^2\gamma\,\sin^2\theta_2
&
2\sin^2\gamma\,\sin\theta_2\cos\theta_2
&
2\sin\gamma\cos\gamma\,\sin\theta_2
\\[1mm]
-2\sin^2\gamma\,\sin\theta_2\cos\theta_2
&
2\sin^2\gamma\,\cos^2\theta_2-1
&
2\sin\gamma\cos\gamma\,\cos\theta_2
\\[1mm]
-2\sin\gamma\cos\gamma\,\sin\theta_2
&
2\sin\gamma\cos\gamma\,\cos\theta_2
&
2\cos^2\gamma-1
\end{pmatrix},
\end{equation}
and
\begin{equation}\label{eq:Gm-reduced-sec2}
G_m=
\begin{pmatrix}
\cos(2\theta_2) & \sin(2\theta_2) & 0\\
-\sin(2\theta_2) & \cos(2\theta_2) & 0\\
0 & 0 & 1
\end{pmatrix}.
\end{equation}
Hence the reduced partial-search dynamics is realized inside $O(3)$. The noncommutativity of $G_n$ and $G_m$ is precisely what makes the ordering problem nontrivial. A schematic representation of this reduced three-dimensional GRK trajectory is shown in Fig.~\ref{fig:grk-trajectory}.

For the asymptotic analysis developed later, we pass to the large-block continuous limit. Since $\theta_1\sim 1/\sqrt N$ and $\theta_2\sim 1/\sqrt b$, both elementary Grover operators become near-identity rotations, and one may regard them as exponentials of infinitesimal generators. Writing $s=\sin\gamma$ and $c=\cos\gamma$, the corresponding generators for $G_n$ and $G_m$ are
\begin{equation}\label{eq:XY-generators-sec2}
X=
\begin{pmatrix}
0 & s^2 & sc\\
-s^2 & 0 & 0\\
-sc & 0 & 0
\end{pmatrix},
\qquad
Y=
\begin{pmatrix}
0 & 1 & 0\\
-1 & 0 & 0\\
0 & 0 & 0
\end{pmatrix}.
\end{equation}
These generators are skew-symmetric, namely $X^\top=-X$ and $Y^\top=-Y$, so $X,Y\in \mathfrak{so}(3)$. The partial-search problem is therefore recast as a time-optimal switching problem between two fixed generators on the reduced state space. 


Note that $G_n$ has determinant $-1$. Thus the natural formulation lives in $O(3)$ rather than in $SO(3)$. However, for the asymptotic viewpoint only the connected, infinitesimal part of the motion is relevant. Indeed, in the large-block limit one has $\theta_1,\theta_2\to 0$, so each elementary Grover step is a small orthogonal transformation close to the identity after removal of an inessential overall sign convention. Consequently, the local behavior of the dynamics is governed by tangent vectors at the identity, namely by the Lie algebra $\mathfrak{so}(3)$ of skew-symmetric generators. For this reason the relevant continuous limit is formulated in terms of the generators $X,Y\in\mathfrak{so}(3)$.

\begin{figure}[t]
\centering
\includegraphics[width=0.6\linewidth]{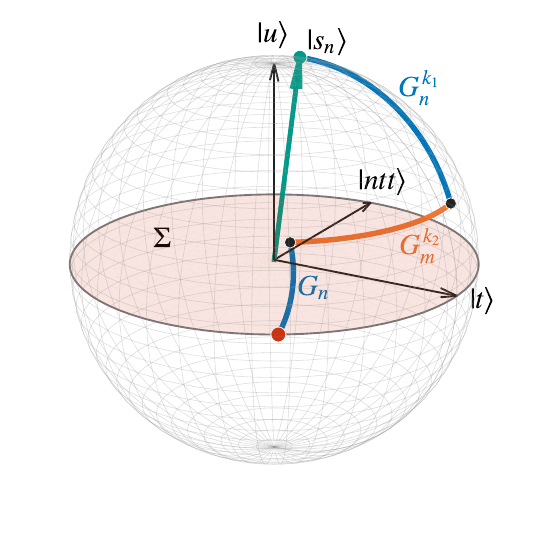}
\caption{Schematic trajectory of the GRK operator in the reduced basis \(\{\ket{t},\ket{ntt},\ket{u}\}\). The blue arcs represent the global Grover stages \(G_n^{k_1}\) and the final \(G_n\), while the orange arc represents the local Grover stage \(G_m^{k_2}\). The shaded disk denotes the terminal plane \(\Sigma=\{\ket{\psi}:\langle u|\psi\rangle=0\}\), and the green arrow indicates the initial state \(\ket{s_n}\).}
\label{fig:grk-trajectory}
\end{figure}

\section{\label{sec:PMP}Time-optimal control problem and Pontryagin maximum principle}

In Sec. \ref{subsec:three}, the partial-search problem was reduced to the three-dimensional $\mathcal H_{\mathrm{red}}=\mathrm{span}\{|t\rangle,|ntt\rangle,|u\rangle\}$ real space with terminal condition \(\langle u|\psi\rangle=0\). In this section, we first reformulate the
search for an optimal sequence of global and local Grover stages as a time-optimal control
problem on \(\mathcal H_{\mathrm{red}}\) in Sec. \ref{subsec:time_optimal}. We then introduce
the PMP in Sec. \ref{subsec:PMP} as the basic tool for
analyzing the optimal switching structure. 

\subsection{\label{subsec:time_optimal}Time-optimal control formulation of partial search}

Our goal is not merely to optimize the iteration numbers within a prescribed ordering, but to
determine the optimal ordering of global and local search stages among all admissible alternating
sequences. In the large-block asymptotic regime, this discrete optimization problem admits a
natural continuous-limit formulation.

Recall that the global and local Grover operators \(G_n\) and \(G_m\) are associated with the
skew-symmetric generators \(X\) and \(Y\) defined in Eq.
\eqref{eq:XY-generators-sec2}. Since $X,Y\in\mathfrak{so}(3)$, both generate norm-preserving rotations
on \(\mathcal H_{\mathrm{red}}\). The partial-search task is to steer the initial state
\(|s_n\rangle\) to the terminal plane $\Sigma$ defined in Eq. \eqref{eq:Sigma-sec2}, with minimum total query cost. Geometrically, this is a minimum-time hitting problem for the plane \(\Sigma\). In the continuous asymptotic description, the problem becomes the time-optimal control system
\begin{equation}
\dot\psi(t)=A(t)\psi(t),\qquad A(t)\in\{X,Y\},
\label{eq:control-system-A}
\end{equation}
with boundary conditions
\begin{equation}
\psi(0)=|s_n\rangle,\qquad \psi(T)\in\Sigma,
\label{eq:control-boundary-PMP}
\end{equation}
and objective $\min T$. 

Equivalently, one may introduce binary controls \(u_X(t),u_Y(t)\in\{0,1\}\) satisfying
\(u_X(t)+u_Y(t)=1\), and write
\begin{equation}
\dot\psi(t)=u_X(t)X\psi(t)+u_Y(t)Y\psi(t).
\label{eq:binary-control-system}
\end{equation}
We restrict attention to measurable piecewise-constant controls, so that admissible trajectories
are concatenations of \(X\)- and \(Y\)-arcs.

A piecewise-constant control corresponds precisely to a concatenation of global and local search
operators. The elapsed time spent on an \(X\)-arc or a \(Y\)-arc is the continuous-limit counterpart
of the number of global or local Grover queries, respectively. Thus minimizing the total oracle
count is, in the asymptotic regime \(N\to\infty\) and \(b\to\infty\), equivalent to minimizing
the total control time.

Since the control set consists of two isolated points, the PMP leads naturally to a bang--bang
description in terms of concatenations of \(X\)- and \(Y\)-arcs, except possibly for singular
regimes to be excluded later. The main structural question is therefore the following:
\emph{among all bang--bang trajectories steering \(|s_n\rangle\) to \(\Sigma\), which switching
patterns can be time-optimal?} To address this question, we now invoke the PMP.

\subsection{\label{subsec:PMP}Pontryagin maximum principle and the switching function}

The PMP gives a first-order necessary condition for optimality in
time-minimization problems \cite{Pontryagin1962,Liberzon2012,AgrachevSachkov2004}. For the
present two-generator problem, it reduces the choice between the two admissible motions
\(X\psi\) and \(Y\psi\) to the sign of a single scalar quantity, the switching function.

Let \(p(t)\in(\mathcal H_{\mathrm{red}})^*\cong\mathbb R^3\) be the costate. For a time-optimal
problem, the normal PMP function is
\begin{equation}
H(p,\psi;A)=\langle p,A\psi\rangle-1,\qquad A\in\{X,Y\}.
\label{eq:normal-Hamiltonian}
\end{equation}
The constant term \(-1\) reflects the fact that time is spent at unit rate. We restrict attention
to normal extremals, that is, to the nondegenerate PMP branch in which the running time enters
the PMP function with a nonzero multiplier. These are the extremals relevant to the bang--bang
analysis developed below. Here and below, normal refers to this nonzero-multiplier PMP branch, whereas regular refers to the switching structure: a regular normal extremal has transversal switching points, \(\Phi=0\) and \(\dot\Phi=\phi_2\neq0\), and contains no interval on which \(\Phi\equiv0\).

According to the PMP, along a normal time-optimal trajectory there exists a nonzero costate
\(p(t)\) such that
\begin{equation}
\dot\psi(t)=A(t)\psi(t),\qquad
\dot p(t)=-A(t)^{\top}p(t),
\label{eq:state-costate-system}
\end{equation}
and, for almost every \(t\), the chosen control \(A(t)\in\{X,Y\}\) maximizes the PMP function:
\begin{equation}
H(p(t),\psi(t);A(t))
=
\max_{B\in\{X,Y\}} H(p(t),\psi(t);B).
\label{eq:H-max-condition}
\end{equation}

Because only two controls are available, the relevant comparison is simply the difference
between the two values of this PMP function. This leads to the following definition.

\begin{definition}
The switching function is
\begin{equation}
\Phi(t):=\langle p(t),(X-Y)\psi(t)\rangle.
\label{eq:switching-function}
\end{equation}
\end{definition}
\noindent The sign of \(\Phi\) determines which of the two admissible generators is instantaneously
preferred from the viewpoint of time-optimality.
\begin{lemma}\label{lem:PMP-sign-rule}
Let \((\psi(\cdot),p(\cdot),A(\cdot))\) be a regular normal extremal for the time-optimal
problem \eqref{eq:control-system-A}--\eqref{eq:control-boundary-PMP}. Then:
\begin{enumerate}
\item if \(\Phi(t)>0\), the maximizing control is \(A(t)=X\);
\item if \(\Phi(t)<0\), the maximizing control is \(A(t)=Y\);
\item a switching between \(X\) and \(Y\) can occur only at times when \(\Phi(t)=0\).
\end{enumerate}
\end{lemma}

\begin{proof}
From Eq. \eqref{eq:normal-Hamiltonian}, we have
\begin{equation}
H(p,\psi;X)-H(p,\psi;Y)
=
\bigl(\langle p,X\psi\rangle-1\bigr)-\bigl(\langle p,Y\psi\rangle-1\bigr)
=
\langle p,(X-Y)\psi\rangle
=
\Phi.
\end{equation}
Hence $\Phi>0$ corresponds to $H(p,\psi;X)>H(p,\psi;Y)$ ($\Phi<0$ corresponds to $H(p,\psi;Y)>H(p,\psi;X)$). By the maximization condition \eqref{eq:H-max-condition}, the maximizing control is \(X\)
when \(\Phi>0\) and \(Y\) when \(\Phi<0\). A switching can therefore occur only at times when
the two PMP function values coincide, equivalently when \(\Phi=0\).
\end{proof}

Lemma \ref{lem:PMP-sign-rule} is the key structural consequence of the PMP for the present
problem. Once the switching function is known, the control is known. Hence the entire
switching pattern of a regular normal extremal is encoded in the zero set of \(\Phi\). The remainder of the analysis is therefore devoted to the evolution of \(\Phi\) along \(X\)- and
\(Y\)-arcs. Since the control is recovered from the sign of \(\Phi\), the zero structure of the
switching function determines the admissible optimal switching patterns.

It is worth emphasizing that the explicit form of the costate \(p(t)\) is not needed for the structural analysis. Indeed, the PMP is used here only through the scalar quantities \(\langle p(t),A\psi(t)\rangle\) with \(A\in\{X,Y\}\), rather than through a closed-form solution for \(p(t)\) itself. The role of the costate is therefore auxiliary. It encodes the first-order sensitivity of the optimal hitting time, but all information relevant to the switching structure is contained in the induced switching function and its reduced descendants. What matters is only that a nonzero costate exists along a normal extremal and satisfies the adjoint equation, since this is sufficient to derive the closed reduced dynamics governing the switching pattern.

\section{\label{sec:switching_dynamics}Reduced switching dynamics}

In this section, we derive a closed reduced system governing the switching function and integrate it explicitly on \(X\)- and \(Y\)-arcs. We first introduce the reduced switching variables and prove the Lie closure of the corresponding observables in Sec. \ref{subsec:reduced_switching_variables}. We then obtain the reduced dynamics on \(X\)- and \(Y\)-arcs in Sec. \ref{subsec:reduced_dynamics}. Finally, in Sec. \ref{subsec:reduced_evolution}, we derive the explicit propagation law between switching points.

\subsection{\label{subsec:reduced_switching_variables}Reduced switching variables and Lie closure}

We now introduce the reduced variables that encode the switching function and its first descendants under the \(X\)- and \(Y\)-flows. The point of this construction is that the switching analysis closes on a finite-dimensional space of observables, hence the term ``reduced'', and no explicit formula for the costate is required.

Recall that along a normal extremal the state and costate satisfy Eq. \eqref{eq:state-costate-system} and the switching function is given by Eq. \eqref{eq:switching-function}. It is therefore natural to consider matrix observables of the form \(\langle p,F\psi\rangle\).

\begin{lemma}\label{lem:reduced-observable-transport}
Let \(F\) be any fixed \(3\times 3\) real matrix. Along a pure \(A\)-arc, with \(A\in\{X,Y\}\) constant, one has
\begin{equation}
\frac{d}{dt}\langle p,F\psi\rangle=\langle p,[F,A]\psi\rangle,
\end{equation}
where \([F,A]=FA-AF\).
\end{lemma}

\begin{proof}
Using \(\dot\psi=A\psi\) and \(\dot p=-A^\top p\), we have
\begin{equation}
\frac{d}{dt}\langle p,F\psi\rangle
=\langle \dot p,F\psi\rangle+\langle p,F\dot\psi\rangle
=\langle -A^\top p,F\psi\rangle+\langle p,FA\psi\rangle.
\end{equation}
Since \(\langle -A^\top p,F\psi\rangle=-\langle p,AF\psi\rangle\), it follows that
\begin{equation}
\frac{d}{dt}\langle p,F\psi\rangle
=\langle p,(FA-AF)\psi\rangle
=\langle p,[F,A]\psi\rangle,
\end{equation}
as claimed.
\end{proof}

The above lemma shows that the evolution of any observable \(\langle p,F\psi\rangle\) is governed by commutators with the active generator. We now choose a basis adapted to the switching function and to the dynamics generated by \(X\) and \(Y\) in Eq. \eqref{eq:XY-generators-sec2}.

\begin{definition}\label{def:reduced-switching-vars}
Define
\begin{equation}
F_1:=X-Y,\qquad F_2:=[X,Y],\qquad F_3:=[Y,[X,Y]].
\end{equation}
The associated reduced switching variables are
\begin{equation}
\phi_1:=\langle p,F_1\psi\rangle,\qquad
\phi_2:=\langle p,F_2\psi\rangle,\qquad
\phi_3:=\langle p,F_3\psi\rangle.
\end{equation}
In particular, \(\phi_1=\Phi\).
\end{definition}

The apparently asymmetric definition of \(F_3\) is dictated by the closure of the reduced switching dynamics: \([X,[X,Y]]=\sin^2\gamma\,(X-Y)\) is already proportional to \(F_1\), whereas differentiating \(F_2=[X,Y]\) along a \(Y\)-arc produces the only new independent direction \(F_3=[Y,[X,Y]]\), and another choice such as \([(X+Y),[X,Y]]\) would merely mix this direction with \(F_1\). The relevance of this choice is that the span of \(\{F_1,F_2,F_3\}\) is invariant under commutation with \(X\) and \(Y\).

\begin{lemma}\label{lem:XY-Lie-closure}
The matrices \(F_1,F_2,F_3\) satisfy the exact identities
\begin{align}
[F_1,X]&=F_2, \label{eq:F1X}\\
[F_1,Y]&=F_2, \label{eq:F1Y}\\
[F_2,X]&=-\sin^2\gamma\,F_1, \label{eq:F2X}\\
[F_2,Y]&=-F_3, \label{eq:F2Y}\\
[F_3,X]&=\sin^2\gamma\,F_2, \label{eq:F3X}\\
[F_3,Y]&=\phantom{-}F_2. \label{eq:F3Y}
\end{align}
Consequently, the space \(\operatorname{span}\{F_1,F_2,F_3\}\) is invariant under the adjoint actions \(\operatorname{ad}_X\) and \(\operatorname{ad}_Y\).
\end{lemma}

\begin{proof}
The identities are immediate from the definitions of \(F_1,F_2,F_3\) together with the commutator relations
\begin{equation}
[X,[X,Y]]=\sin^2\gamma\,(X-Y),\qquad
[Y,[X,Y]]=F_3,
\end{equation}
and
\begin{equation}
[X,F_3]=-\sin^2\gamma\,[X,Y],\qquad
[Y,F_3]=- [X,Y].
\end{equation}
It is straightforward to verify that
\begin{equation}
[F_1,X]=[X-Y,X]=-[Y,X]=[X,Y]=F_2,
\end{equation}
and similarly
\begin{equation}
[F_1,Y]=[X-Y,Y]=[X,Y]=F_2.
\end{equation}
Next, we have
\begin{equation}
[F_2,X]=[[X,Y],X]=-[X,[X,Y]]
=-\sin^2\gamma\,(X-Y)=-\sin^2\gamma\,F_1,
\end{equation}
while by definition of \(F_3\),
\begin{equation}
[F_2,Y]=[[X,Y],Y]=-[Y,[X,Y]]=-F_3.
\end{equation}
Finally,
\begin{equation}
[F_3,X]=-[X,F_3]=\sin^2\gamma\,[X,Y]=\sin^2\gamma\,F_2,
\qquad
[F_3,Y]=-[Y,F_3]=[X,Y]=F_2.
\end{equation}
Thus commutation with either \(X\) or \(Y\) maps \(\operatorname{span}\{F_1,F_2,F_3\}\) into itself.
\end{proof}

Lemma \ref{lem:reduced-observable-transport} together with Lemma \ref{lem:XY-Lie-closure} provides the basic simplification for the switching analysis. Once one tracks the triple \((\phi_1,\phi_2,\phi_3)\), no further commutator variables are generated: the switching function and its descendants evolve inside a closed three-dimensional linear system. In other words, the PMP dynamics induces a finite-dimensional reduced switching dynamics, and this closed subsystem is the only part of the state-costate pair needed in the subsequent analysis of switching points, singular arcs, and compression.

\subsection{\label{subsec:reduced_dynamics}Reduced dynamics on \(X\)- and \(Y\)-arcs}

We now use the closed commutator structure from Lemma~\ref{lem:reduced-observable-transport} and Lemma~\ref{lem:XY-Lie-closure} to write the reduced dynamics explicitly on each bang arc.

\begin{lemma}\label{lem:reduced-X-arc}
Along a pure \(X\)-arc, the reduced variables satisfy
\begin{equation}
\dot\phi_1=\phi_2,\qquad
\dot\phi_2=-\sin^2\gamma\,\phi_1,\qquad
\dot\phi_3=\sin^2\gamma\,\phi_2.
\label{eq:reduced-X-system}
\end{equation}
Equivalently, \(\phi_1\) solves the harmonic equation
\begin{equation}
\ddot\phi_1+\sin^2\gamma\,\phi_1=0.
\label{eq:phi1-X-harmonic}
\end{equation}
\end{lemma}

\begin{proof}
On an \(X\)-arc, Lemma~\ref{lem:reduced-observable-transport} gives
\begin{equation}
\dot\phi_j=\langle p,[F_j,X]\psi\rangle,\qquad j=1,2,3.
\end{equation}
Using Lemma~\ref{lem:XY-Lie-closure}, we obtain
\begin{equation}
[F_1,X]=F_2,\qquad [F_2,X]=-\sin^2\gamma\,F_1,\qquad [F_3,X]=\sin^2\gamma\,F_2.
\end{equation}
Therefore
\begin{equation}
\dot\phi_1=\phi_2,\qquad
\dot\phi_2=-\sin^2\gamma\,\phi_1,\qquad
\dot\phi_3=\sin^2\gamma\,\phi_2.
\end{equation}
Differentiating the first identity and substituting the second yields
\begin{equation}
\ddot\phi_1=\dot\phi_2=-\sin^2\gamma\,\phi_1,
\end{equation}
which is Eq. \eqref{eq:phi1-X-harmonic}.
\end{proof}

Thus, on an \(X\)-arc, the switching function \(\phi_1=\Phi\) oscillates harmonically with angular frequency \(\sin\gamma\). In particular, its zeros are isolated unless \(\phi_1\equiv 0\), a possibility to be excluded in Sec. \ref{subsec:singular_arcs}. The corresponding reduced system on a \(Y\)-arc is as follows.

\begin{lemma}\label{lem:reduced-Y-arc}
Along a pure \(Y\)-arc, the reduced variables satisfy
\begin{equation}
\dot\phi_1=\phi_2,\qquad
\dot\phi_2=-\phi_3,\qquad
\dot\phi_3=\phi_2.
\label{eq:reduced-Y-system}
\end{equation}
Equivalently, \(\phi_2\) and \(\phi_3\) satisfy \(\ddot\phi_2=-\phi_2\) and \(\ddot\phi_3=-\phi_3\), while \(\phi_1\) is obtained by one further integration of \(\dot\phi_1=\phi_2\).
\end{lemma}

\begin{proof}
On a \(Y\)-arc, Lemma~\ref{lem:reduced-observable-transport} gives
\begin{equation}
\dot\phi_j=\langle p,[F_j,Y]\psi\rangle,\qquad j=1,2,3.
\end{equation}
Using Lemma~\ref{lem:XY-Lie-closure}, we have
\begin{equation}
[F_1,Y]=F_2,\qquad [F_2,Y]=-F_3,\qquad [F_3,Y]=F_2.
\end{equation}
Hence
\begin{equation}
\dot\phi_1=\phi_2,\qquad
\dot\phi_2=-\phi_3,\qquad
\dot\phi_3=\phi_2.
\end{equation}
Differentiating the last two equations gives \(\ddot\phi_2=-\dot\phi_3=-\phi_2\) and \(\ddot\phi_3=\dot\phi_2=-\phi_3\).
\end{proof}

The above two lemmas show that the reduced switching dynamics is elementary on each bang arc: it is elliptic on both \(X\)- and \(Y\)-arcs, although with different frequencies and with different variables playing the primary role. This distinction is the starting point for the explicit switching-point propagation derived in the next subsection. The different frequencies come from the original generators rather than from the choice of reduced variables: \(X\) contains \(s=\sin\gamma=1/\sqrt K\) because a global step couples the target block to the non-target blocks, whereas \(Y\) acts locally within each block and, after measuring the local evolution time in units of \(\sqrt b\) local Grover queries, contains no dependence on the number of blocks \(K\).

\subsection{\label{subsec:reduced_evolution}Reduced evolution between switching points}

We now integrate the reduced systems obtained in the previous subsection. Since a switching can occur only when \(\phi_1=\Phi=0\), the reduced data at a switching point are naturally parametrized by \((\phi_1,\phi_2,\phi_3)(0)=(0,a,b)\), where \(a=\phi_2(0)\) and \(b=\phi_3(0)\).

\begin{lemma}\label{lem:X-arc-propagation}
Suppose an \(X\)-arc starts at a switching point with reduced initial data \((\phi_1,\phi_2,\phi_3)(0)=(0,a,b)\). Then along the arc one has
\begin{equation}
\phi_1(t)=\frac{a}{\sin\gamma}\sin\bigl((\sin\gamma)t\bigr),\quad
\phi_2(t)=a\cos\bigl((\sin\gamma)t\bigr),\quad
\phi_3(t)=b+a\sin\gamma\,\sin\bigl((\sin\gamma)t\bigr).
\label{eq:X-arc-explicit}
\end{equation}
In particular, the shortest positive time at which the next switching can occur is
\begin{equation}
\tau_X=\frac{\pi}{\sin\gamma},
\label{eq:tauX}
\end{equation}
and at that time one has
\begin{equation}
(\phi_1,\phi_2,\phi_3)(\tau_X)=(0,-a,b).
\label{eq:X-return-map}
\end{equation}
\end{lemma}

\begin{proof}
Along an \(X\)-arc, Lemma~\ref{lem:reduced-X-arc} gives
\begin{equation}
\dot\phi_1=\phi_2,\qquad
\dot\phi_2=-\sin^2\gamma\,\phi_1,\qquad
\dot\phi_3=\sin^2\gamma\,\phi_2.
\end{equation}
Hence \(\phi_1\) satisfies \(\ddot\phi_1+\sin^2\gamma\,\phi_1=0\), with initial conditions \(\phi_1(0)=0\) and \(\dot\phi_1(0)=\phi_2(0)=a\). Therefore
\begin{equation}
\phi_1(t)=\frac{a}{\sin\gamma}\sin\bigl((\sin\gamma)t\bigr),
\end{equation}
and then
\begin{equation}
\phi_2(t)=\dot\phi_1(t)=a\cos\bigl((\sin\gamma)t\bigr).
\end{equation}
Finally, since \(\dot\phi_3=\sin^2\gamma\,\phi_2\) and \(\phi_3(0)=b\), we obtain
\begin{equation}
\phi_3(t)=b+\int_0^t \sin^2\gamma\,\phi_2(\tau)\,d\tau
=b+a\sin\gamma\,\sin\bigl((\sin\gamma)t\bigr).
\end{equation}
This proves Eq. \eqref{eq:X-arc-explicit}.

Because \(\phi_1=\Phi\), the next switching time is the first positive zero of \(\phi_1\), namely the first positive zero of \(\sin((\sin\gamma)t)\), which is \(t=\pi/\sin\gamma\). Evaluating \eqref{eq:X-arc-explicit} at \(t=\tau_X\) gives
\begin{equation}
\phi_1(\tau_X)=0,\qquad
\phi_2(\tau_X)=a\cos\pi=-a,\qquad
\phi_3(\tau_X)=b+a\sin\gamma\,\sin\pi=b,
\end{equation}
which is exactly Eq. \eqref{eq:X-return-map}.
\end{proof}

\begin{lemma}\label{lem:Y-arc-propagation}
Suppose a \(Y\)-arc starts at a switching point with reduced initial data \((\phi_1,\phi_2,\phi_3)(0)=(0,a,b)\). Then along the arc one has
\begin{equation}
\phi_1(t)=a\sin t+b(\cos t-1),\quad
\phi_2(t)=a\cos t-b\sin t,\quad
\phi_3(t)=a\sin t+b\cos t.
\label{eq:Y-arc-explicit}
\end{equation}
The next switching time is determined by \(\phi_1(t)=0\). Besides the trivial root \(t=0\), the shortest positive switching time \(\tau_Y\in(0,2\pi)\) is characterized by
\begin{equation}
a\cos\frac{\tau_Y}{2}-b\sin\frac{\tau_Y}{2}=0,
\label{eq:Y-half-angle}
\end{equation}
equivalently, when \(b\neq 0\),
\begin{equation}
\tan\frac{\tau_Y}{2}=\frac{a}{b}.
\label{eq:Y-tan-half}
\end{equation}
At that switching time one has
\begin{equation}
(\phi_1,\phi_2,\phi_3)(\tau_Y)=(0,-a,b).
\label{eq:Y-return-map}
\end{equation}
\end{lemma}

\begin{proof}
Along a \(Y\)-arc, Lemma~\ref{lem:reduced-Y-arc} gives
\begin{equation}
\dot\phi_1=\phi_2,\qquad
\dot\phi_2=-\phi_3,\qquad
\dot\phi_3=\phi_2.
\end{equation}
The last two equations form the standard planar rotation system. With initial data \((\phi_2,\phi_3)(0)=(a,b)\), we obtain
\begin{equation}
\phi_2(t)=a\cos t-b\sin t,\qquad
\phi_3(t)=a\sin t+b\cos t.
\end{equation}
Integrating \(\dot\phi_1=\phi_2\) and using \(\phi_1(0)=0\) yields
\begin{equation}
\phi_1(t)=a\sin t+b(\cos t-1),
\end{equation}
which proves Eq. \eqref{eq:Y-arc-explicit}.

To determine the next switching time, we impose \(\phi_1(t)=0\). Using \(\sin t=2\sin(t/2)\cos(t/2)\) and \(\cos t-1=-2\sin^2(t/2)\), we rewrite
\begin{equation}
\phi_1(t)=2\sin\frac{t}{2}\left(a\cos\frac{t}{2}-b\sin\frac{t}{2}\right).
\end{equation}
Thus, besides the trivial root \(t=0\), the next switching time is characterized by
\begin{equation}
a\cos\frac{t}{2}-b\sin\frac{t}{2}=0.
\end{equation}
Let \(\tau_Y\in(0,2\pi)\) denote the shortest positive solution. When \(b\neq 0\), this is equivalent to Eq. \eqref{eq:Y-tan-half}.

It remains to evaluate the reduced data at \(t=\tau_Y\). Set \(u=\tau_Y/2\) and \(r=a/b\), so that \(\tan u=r\). Then
\begin{equation}
\cos\tau_Y=\cos(2u)=\frac{1-r^2}{1+r^2},\qquad
\sin\tau_Y=\sin(2u)=\frac{2r}{1+r^2}.
\end{equation}
Substituting these into Eq. \eqref{eq:Y-arc-explicit}, we find
\begin{equation}
\phi_2(\tau_Y)=a\cos\tau_Y-b\sin\tau_Y=-a,\qquad
\phi_3(\tau_Y)=a\sin\tau_Y+b\cos\tau_Y=b.
\end{equation}
Together with \(\phi_1(\tau_Y)=0\), this proves Eq. \eqref{eq:Y-return-map}.
\end{proof}

We summarize the above two propagation lemmas as follows.
\begin{proposition}\label{prop:universal-reduced-switching-map}
Let \((0,a,b)\) be the reduced data at a switching point. Then the shortest switching-to-switching \(X\)-arc and the shortest switching-to-switching \(Y\)-arc induce the same reduced map,
\begin{equation}
(0,a,b)\mapsto (0,-a,b).
\label{eq:universal-switching-map}
\end{equation}
Equivalently, at switching points the reduced variables transform as $a\mapsto -a$ and $b\mapsto b$.
\end{proposition}

Proposition~\ref{prop:universal-reduced-switching-map} is the key summary of the reduced propagation analysis. Although \(X\)- and \(Y\)-arcs have different interior dynamics and generally different durations, their first-return action on the switching surface \(\{\phi_1=0\}\) is identical: the transverse switching velocity \(a=\phi_2\) changes sign, while the residual parameter \(b=\phi_3\) is preserved. This universal switching map underlies the compression arguments developed below.

\section{\label{sec:singluar_arcs}Universal switching map and exclusion of singular arcs}

This section extracts the two structural consequences of the reduced switching dynamics that
will be used repeatedly below. First, we show that the reduced data at a regular switching point determine
the subsequent reduced continuation uniquely in Sec. \ref{subsec:switching_map}. Second, we prove that singular arcs are impossible for normal
time-optimal extremals in Sec. \ref{subsec:singular_arcs}.

\subsection{\label{subsec:switching_map}Universal switching map}

Recall from Proposition \ref{prop:universal-reduced-switching-map} that the shortest switching-to-switching
$X$-arc and the shortest switching-to-switching $Y$-arc induce the same reduced map, namely $(0,a,b)\mapsto (0,-a,b)$. To use this fact in compression arguments, one needs a continuation statement saying that once
the reduced data at a regular switching point are fixed, the later reduced bang--bang evolution
is fixed as well.

\begin{lemma}\label{lem:reduced-ode-uniqueness}
For each $A\in\{X,Y\}$ and each initial data $m_0\in\mathbb{R}^3$, the reduced ODE
\begin{equation}
\dot m = V_A(m),\qquad m(0)=m_0,
\end{equation}
has a unique global solution. Consequently, if two reduced trajectories for the same pure arc
type $A$ agree at one time, then they agree for all times for which both are defined.
\end{lemma}

\begin{proof}
By Lemmas \ref{lem:reduced-X-arc} and \ref{lem:reduced-Y-arc}, the reduced vector fields $V_X$ and $V_Y$ are
linear. Hence they are smooth and globally Lipschitz on $\mathbb{R}^3$. Existence and
uniqueness for the corresponding initial-value problems therefore follow from the standard Picard--Lindel\"of theorem \cite{CoddingtonLevinson1955,Hartman2002}.
\end{proof}

\begin{lemma}\label{lem:reduced-switching-continuation}
Let $m_*=(\phi_1^*,\phi_2^*,\phi_3^*)$ be reduced data at some time of a normal extremal.

\begin{enumerate}
\item If $\phi_1^*>0$, then the future reduced extremal with the initial $m_*$ is uniquely
determined and is the unique solution of
\begin{equation}
\dot m = V_X(m),\qquad m(0)=m_*,
\end{equation}
for as long as $\phi_1$ stays positive.

\item If $\phi_1^*<0$, then the future reduced extremal with the initial $m_*$ is uniquely
determined and is the unique solution of
\begin{equation}
\dot m = V_Y(m),\qquad m(0)=m_*,
\end{equation}
for as long as $\phi_1$ stays negative.

\item If $\phi_1^*=0$ and $\phi_2^*\neq 0$, then $m_*$ is a regular switching point for the
reduced dynamics, and the future reduced bang--bang extremal with the initial $m_*$ is
uniquely determined.
\end{enumerate}
\end{lemma}

\begin{proof}
The first two statements are immediate from the PMP sign rule in Lemma
\ref{lem:PMP-sign-rule}: the sign of $\phi_1=\Phi$ determines the maximizing control,
and once the control is fixed, Lemma \ref{lem:reduced-ode-uniqueness} gives uniqueness of the
reduced trajectory.

For the third case, let $\phi_1^*=0$ and $\phi_2^*\neq 0$. Since
$\dot\phi_1=\phi_2$ on both $X$- and $Y$-arcs by Lemmas \ref{lem:reduced-X-arc} and \ref{lem:reduced-Y-arc},
we have
\begin{equation}
\dot\phi_1(0)=\phi_2^*\neq 0.
\end{equation}
Hence the reduced trajectory crosses the switching surface $\{\phi_1=0\}$ transversely, so
the sign of $\phi_1(t)$ for small $t>0$ is determined by the sign of $\phi_2^*$. Therefore the
outgoing maximizing arc type is uniquely selected. Once that outgoing arc type is fixed, Lemma
\ref{lem:reduced-ode-uniqueness} again gives uniqueness of the subsequent reduced
continuation.
\end{proof}

\begin{lemma}\label{lem:equal-switching-data-equal-future}
If two regular bang--bang extremals pass through switching points with the same reduced
switching data
\begin{equation}
(\phi_1,\phi_2,\phi_3)=(0,a,b),\qquad a\neq 0,
\end{equation}
then their future reduced bang--bang continuations coincide.
\end{lemma}

\begin{proof}
At a regular switching point one has $\phi_1=0$ and $\phi_2=a\neq 0$. Therefore Lemma
\ref{lem:reduced-switching-continuation} applies and uniquely determines the outgoing
reduced bang--bang germ from the data $(0,a,b)$. Hence two extremals with the same reduced
switching data have the same future reduced continuation.
\end{proof}

\begin{proposition}\label{prop:universal-switching-continuation}
At a regular switching point, the reduced switching data determine the future reduced bang--bang
continuation uniquely. Equivalently, the universal switching map $(0,a,b)\mapsto (0,-a,b)$ can be used as a genuine propagation rule on the switching surface: once the reduced data at
the endpoint of a bang are prescribed, the later reduced extremal is fixed.
\end{proposition}

\begin{proof}
The first statement is exactly Lemma \ref{lem:equal-switching-data-equal-future}. The second is an
equivalent reformulation using Proposition \ref{prop:universal-reduced-switching-map}: the reduced first-return
map specifies the next switching datum, and Lemma \ref{lem:reduced-switching-continuation} then
determines the later reduced bang--bang continuation uniquely.
\end{proof}

\subsection{\label{subsec:singular_arcs}Exclusion of singular arcs}

We now show that normal time-optimal extremals cannot contain singular subarcs of positive length in our case. This exclusion is needed because a singular subarc would represent an interval on which the PMP does not select the control through the sign of the switching function, thereby obstructing the reduction of candidate minimizers to a genuinely bang--bang switching structure and, in particular, the compression arguments used later.

\begin{definition}\label{def:singular-arc}
A singular arc is an interval $I\subset[0,T]$ on which the switching function vanishes
identically $\Phi(t)=0$ for all $t\in I$. Equivalently, since $\phi_1=\Phi$, it is an interval on which $\phi_1(t)\equiv 0$.
\end{definition}

\begin{lemma}\label{thm:no-singular-arcs}
A normal time-optimal extremal composed of $X$- and $Y$-arcs contains no singular arc of positive length.
\end{lemma}

\begin{proof}
Assume that $I$ is a singular interval. Then, by definition, we have $\phi_1(t)=0$ with $t\in I$. Differentiating and using $\dot\phi_1=\phi_2$ on both $X$- and $Y$-arcs, given by Lemmas \ref{lem:reduced-X-arc} and \ref{lem:reduced-Y-arc}, we obtain $\phi_2(t)=0$ with $t\in I$. 

We next use the reduced equations to determine which controls are compatible with
$\phi_1\equiv\phi_2\equiv 0$ on $I$. On a $Y$-arc, Lemma \ref{lem:reduced-Y-arc} gives $\dot\phi_2=-\phi_3$. Since $\phi_2\equiv 0$ on $I$, this implies $\phi_3(t)=0$ almost everywhere on the portion of $I$ where the active control is $Y$. Thus, on any such
portion, the full reduced triple is $(\phi_1,\phi_2,\phi_3)=(0,0,0)$. But then the two PMP function values associated with $X$ and $Y$ are equal there, so the control is
not selected by the PMP sign rule. In the present bang--bang analysis, every nontrivial bang subarc of a normal extremal is understood as a regular PMP-selected arc. Therefore such a portion cannot be a genuine regular
bang arc. Hence a singular interval cannot contain a nontrivial $Y$-subarc. Therefore $u_Y(t)=0$ for almost every $t\in I$, so the singular interval must be a pure $X$-arc almost everywhere.

Along a pure $X$-arc, Lemma \ref{lem:reduced-X-arc} gives $\dot\phi_3=\sin^2\gamma\,\phi_2$. Since $\phi_2\equiv 0$ on $I$, it follows that $\dot\phi_3=0$ on $I$. Hence $\phi_3$ is constant on $I$. Together with $\phi_1\equiv 0$ and $\phi_2\equiv 0$, this
shows that the full reduced moment vector $m(t)=(\phi_1(t),\phi_2(t),\phi_3(t))$ is constant on $I$.

Let $t_-<t_+$ be two distinct times in $I$. Then we have $m(t_-)=m(t_+)$. Moreover, since the singular interval is a pure $X$-arc almost everywhere, the outgoing
maximizing arc immediately after $t_-$ and immediately after $t_+$ is of the same type.
Therefore the same reduced ODE with the same initial data governs both future
continuations. By Lemma \ref{lem:reduced-ode-uniqueness}, the future reduced continuations
from $t_-$ and $t_+$ coincide. Hence the whole subarc $[t_-,t_+]$ may be deleted without
changing the later reduced evolution. Deleting it strictly decreases the total time,
contradicting optimality. Therefore a singular arc cannot have positive length.
\end{proof}

The above proof also implies that any singular interval would have to be a pure $X$-arc on which
the reduced moment vector is constant, hence a removable subarc. Thus singular behavior is
incompatible with time-optimality precisely because it produces no nontrivial reduced
propagation while still consuming positive time.

\section{\label{sec:optimality}Asymptotic structural optimality}

In this section, we complete the asymptotic structural optimality proof for the GRK ordering. We first show in Sec. \ref{subsec:compression_three} that interior three-bang blocks of type $YXY$ and $XYX$ can be compressed to single bangs with the same reduced effect and shorter time. We then use this in Sec. \ref{subsec:final_structure} to prove that every regular normal time-optimal extremal has the pattern $XYX$.

\subsection{\label{subsec:compression_three}Compression of three-bang blocks}

We now turn the reduced switching analysis into a genuine compression principle. By Proposition \ref{prop:universal-reduced-switching-map}, the shortest switching-to-switching $X$- and $Y$-arcs induce the same reduced map on the switching surface. By Proposition \ref{prop:universal-switching-continuation}, once the reduced switching data at a regular switching point are fixed, the later reduced bang--bang continuation is fixed as well. Therefore, if a multi-bang block has the same reduced endpoint data as a single bang but takes more time, it can be replaced by that shorter bang without changing the later reduced evolution. Throughout this section, a bang is a maximal interval on which the active generator is fixed. Hence adjacent identical symbols are merged into one longer bang, and non-alternating strings such as \(XXY\) or \(XYY\) are not distinct three-bang blocks but two-bang patterns after this merging.

\begin{theorem}\label{thm:compress-YXY}
An interior three-bang block $Y(\ell_1)\,X(\tau_X)\,Y(\ell_2)$ with switching time $\tau_X=\pi/\sin\gamma$ given by Eq. \eqref{eq:tauX} and $\ell_1,\ell_2\in(0,2\pi)$ on a regular normal bang--bang extremal is reducible to a single switching-to-switching $X$-arc with shorter time.
\end{theorem}

\begin{proof}
Let the incoming switching point have reduced data $(\phi_1,\phi_2,\phi_3)=(0,a,b)$ with $a\neq 0$. By Proposition \ref{prop:universal-reduced-switching-map}, the first $Y$-arc sends $(0,a,b)\mapsto (0,-a,b)$, the intermediate $X$-arc sends $(0,-a,b)\mapsto (0,a,b)$, and the final $Y$-arc sends $(0,a,b)\mapsto (0,-a,b)$. Hence the full $YXY$ block induces the same reduced endpoint map $(0,a,b)\mapsto (0,-a,b)$ as a single switching-to-switching $X$-arc.

Therefore the endpoint reduced switching data of the three-bang block coincide with those of a single $X$-arc issued from the same incoming switching point. Since the endpoint is regular, Proposition \ref{prop:universal-switching-continuation} implies that the later reduced bang--bang continuation is unchanged after replacing the block by that single $X$-arc.

It remains to compare durations. By Lemma \ref{lem:X-arc-propagation}, the switching-to-switching $X$-arc has length $\tau_X=\pi/\sin\gamma$. By Lemma \ref{lem:Y-arc-propagation}, the two $Y$-arc lengths lie in $(0,2\pi)$, and for two consecutive $Y$-arcs separated by one intermediate bang one has $\ell_1+\ell_2=2\pi$. Hence the total length of the $YXY$ block is
\begin{equation}
L_{YXY}=\ell_1+\frac{\pi}{\sin\gamma}+\ell_2
=2\pi+\frac{\pi}{\sin\gamma}
>
\frac{\pi}{\sin\gamma}
=\tau_X.
\end{equation}
So the replacing single $X$-arc is strictly shorter. Therefore an interior $YXY$ block cannot occur on a regular normal time-optimal extremal.
\end{proof}

\begin{theorem}\label{thm:compress-XYX}
An interior three-bang block $X(\tau_X)\,Y(\ell)\,X(\tau_X)$ with switching time $\tau_X=\pi/\sin\gamma$ given by Eq. \eqref{eq:tauX} and $\ell\in(0,2\pi)$ on a regular normal bang--bang extremal is reducible to a single switching-to-switching $Y$-arc with shorter time.
\end{theorem}

\begin{proof}
Let the incoming switching point have reduced data $(\phi_1,\phi_2,\phi_3)=(0,a,b)$ with $a\neq 0$. Applying Proposition \ref{prop:universal-reduced-switching-map} successively, the first $X$-arc sends $(0,a,b)\mapsto (0,-a,b)$, the intermediate $Y$-arc sends $(0,-a,b)\mapsto (0,a,b)$, and the final $X$-arc sends $(0,a,b)\mapsto (0,-a,b)$. Thus the full $XYX$ block induces the same reduced endpoint map $(0,a,b)\mapsto (0,-a,b)$ as a single switching-to-switching $Y$-arc. Hence the endpoint reduced switching data of the three-bang block coincide with those of a single $Y$-arc issued from the same incoming switching point. By Proposition \ref{prop:universal-switching-continuation}, the later reduced bang--bang continuation is unchanged after replacing the block by that single $Y$-arc.

Comparing durations, the total length of the $XYX$ block is
\begin{equation}
L_{XYX}=\frac{\pi}{\sin\gamma}+\ell+\frac{\pi}{\sin\gamma}
=\frac{2\pi}{\sin\gamma}+\ell,
\end{equation}
whereas the replacing $Y$-arc has length $\ell\in(0,2\pi)$ by Lemma \ref{lem:Y-arc-propagation}. Since $2\pi/\sin\gamma+\ell>\ell$, the three-bang block is strictly longer. Therefore an interior $XYX$ block cannot occur on a regular normal time-optimal extremal.
\end{proof}

\subsection{\label{subsec:final_structure}At most two switchings and asymptotic structural optimality}

In \ref{app:endpoint}, we show that the optimal extremal must start with $X$ and end with $X$. Combining this with the compression results of Sec. \ref{subsec:compression_three}, we obtain the unique nontrivial switching structure of a regular normal time-optimal extremal.

\begin{theorem}\label{thm:at-most-two-switchings}
Any regular normal time-optimal bang--bang extremal has at most two switchings. Equivalently, it contains at most three bangs. In particular, the only nontrivial optimal pattern is $XYX$.
\end{theorem}

\begin{proof}
Assume for contradiction that a regular normal time-optimal bang--bang extremal has at least three switchings. Since the control takes values in $\{X,Y\}$, the corresponding bang sequence is alternating and has length at least four. Therefore it contains a consecutive three-bang block, necessarily of one of the two forms $YXY$ and $XYX$. If the full sequence has more than three bangs, then such a consecutive three-bang block can be chosen to be interior. But Theorem \ref{thm:compress-YXY} excludes interior $YXY$ blocks on a regular normal time-optimal extremal, and Theorem \ref{thm:compress-XYX} excludes interior $XYX$ blocks. This contradiction proves that no regular normal time-optimal bang--bang extremal can have more than two switchings.

Hence any such extremal has at most three bangs. In \ref{app:endpoint}, we show that the first bang must be $X$ and the last bang must also be $X$. Therefore the only possible bang patterns are $X$ and $XYX$. The one-bang pattern $X$ is only a degenerate purely global case. In the generic partial-search setting, the nontrivial regular time-optimal structure is therefore $XYX$.
\end{proof}

Theorem \ref{thm:at-most-two-switchings} is the structural core of the argument. Within the class of regular normal time-optimal extremals in the asymptotic continuous-limit problem, optimality forces the global--local--global arrangement and excludes all longer alternating sequences. It remains only to optimize the arc lengths within this $XYX$ family.

\begin{theorem}\label{thm:GRK-structural-optimality}
Within the asymptotic continuous-limit formulation of partial quantum search, the optimal operator ordering is global--local--global, namely $G_n\,G_m^{k_2}\,G_n^{k_1}$. With the optimal GRK parameters given by Eq. \eqref{eq:GRK-optimal-parameters-sec2}, this yields an asymptotically optimal partial-search algorithm.
\end{theorem}

\begin{proof}
Theorem \ref{thm:at-most-two-switchings} shows that every regular normal time-optimal extremal has the pattern $XYX$. The $X$-arc corresponds to a global Grover stage and the $Y$-arc to a local Grover stage. Therefore the only asymptotically optimal operator ordering is global--local--global, namely $G_n\,G_m^{k_2}\,G_n^{k_1}$.

It remains to optimize the durations within this fixed ordering. This is exactly the standard GRK optimization recalled in Sec. \ref{subsec:GRK}. The optimal parameters are given by \cite{Korepin2005OptimizationPartialSearch}, namely by Eq. \eqref{eq:GRK-optimal-parameters-sec2}. Hence the GRK operator $G_n\,G_m^{k_2}\,G_n^{k_1}$ is asymptotically optimal.
\end{proof}

A discrete parity issue remains in passing from the reduced continuous-limit control problem to the original operator sequence built from $G_n$ and $G_m$. In the three-dimensional
reduction, $G_m$ is orientation-preserving, whereas $G_n$ has $\det G_n=-1$ and hence lies
in $O(3)\setminus SO(3)$. The continuous-limit formulation used here is built from the
infinitesimal generators $X,Y\in\mathfrak{so}(3)$ and therefore describes only the connected,
orientation-preserving part of the dynamics. Consequently, the PMP argument proves the
optimality of the $X$ and $XYX$ structures within the asymptotic continuous-limit problem,
but it does not by itself distinguish the discrete parity sectors associated with an odd or even
total number of global Grover steps.

This does not affect the leading asymptotic query count proved here. Indeed, the parity of the
global stage contributes only a discrete endpoint correction of order $\mathcal O(1)$ when continuous
arc lengths are converted to integer iteration numbers. Hence the present argument establishes
the global--local--global structure and the GRK parameters as asymptotically optimal up to
such $\mathcal O(1)$ parity corrections. This is consistent with the group-theoretic formulation of
partial search, where the reflection part of the global operator is separated from the connected
$SO(3)$ dynamics \cite{KorepinVallilo2006GroupTheoreticalPartialSearch}, and with the discrete large-block analysis of
global--local--global sequences in \cite{KorepinLiao2006QuestFastPartialSearch}, which shows that the nontrivial
optimal discrete realization occurs in the odd-parity sector and is realized by the GRK-type final global step. 


\section{\label{sec:conclusion}Conclusion and discussion}

We have reformulated quantum partial search in the large-block regime as a time-optimal switching problem on the reduced three-dimensional real space, with two admissible generators \(X\) and \(Y\) corresponding to the continuous limits of the global and local Grover operators. In this formulation, the ordering problem becomes a geometric control problem. By combining the PMP with the closed reduced switching dynamics, we proved that regular normal time-optimal extremals contain no singular subarcs of positive length and no compressible interior three-bang blocks. It follows that any such extremal has at most two switchings. The only nontrivial time-optimal pattern is therefore \(XYX\). Since \(X\) and \(Y\) represent global and local Grover stages, respectively, this gives the global--local--global structure as the unique nontrivial asymptotically optimal ordering. Since GRK is optimal in the global--local--global structure, we confirm that the GRK algorithm is asymptotically optimal in the large-block regime.

A related but distinct optimization criterion is the expected number of oracle queries under a restart strategy: the algorithm may be stopped before near-unit success probability and repeated after failure \cite{boyer1996tight}. Such a strategy has recently been analyzed for GRK partial search in Ref.~\cite{JiangWangZhangKorepin2026ExactBoundsPartialSearch}. Within the GRK family, the minimal expected number of oracle queries is \(E_{\min}\simeq 0.69\sqrt N-0.4054\sqrt b\) for \(m\le n/2\). This expected-query optimization is therefore an analogue of the punctuated Grover search, while the present paper focuses on the bang--bang structural optimality of the GRK ordering.

Recent work also shows that the quantum partial search algorithm based on the GRK operator is not efficient, compared with the classical case, when \(m>n/2\) \cite{JiangWangZhangKorepin2026ExactBoundsPartialSearch}. In the present setting, however, we restrict ourselves to the standard partial diffusion operator used in partial search. Whether alternative diffusion operators can offer a further advantage in the regime \(m>n/2\) remains an open question. In this work we use the oracle query count as the performance metric. On actual quantum hardware, however, circuit depth is often the more relevant resource \cite{Preskill2018quantumcomputingin}. Depth optimization for full quantum search has been investigated in recent years \cite{Zhang2020DepthOptimizationQuantumSearch,Brianski2021IntroducingStructure,Campos2024DepthScalingQAOA}. One may therefore expect GRK to remain optimal in depth only in the regime where the oracle depth dominates the total circuit cost. Otherwise, the design of depth-efficient partial search algorithms becomes the more practically relevant problem.

\section*{Data availability statement}

No new data were created or analysed in this study.

\ack

This work was supported by the NSFC (Grants No.12305028, No.12275215, and No.12247103), and the Youth Innovation Team of Shaanxi Universities. KZ is supported by the China Postdoctoral Science Foundation under Grant Number 2025M773421, Shaanxi Province Postdoctoral Science Foundation under Grant Number 2025BSHYDZZ017, and Scientific Research Program Funded by Education Department of Shaanxi Provincial Government (Program No.24JP186). VK is funded by the U.S. Department of Energy, Office of Science, National Quantum Information Science Research Centers, Co-Design Center for Quantum Advantage (C2QA) under Contract No. DE-SC0012704.

\appendix

\section{\label{app:endpoint}Endpoint arc lemmas}

In this appendix we collect two endpoint lemmas for the reduced continuous-limit control problem of Sec.~\ref{subsec:time_optimal}.

\begin{lemma}\label{lem:last-arc-global-app}
Any nontrivial time-optimal trajectory reaching \(\Sigma\) must end with an \(X\)-arc.
\end{lemma}

\begin{proof}
Since
\begin{equation}
Y=
\begin{pmatrix}
0 & 1 & 0\\
-1 & 0 & 0\\
0 & 0 & 0
\end{pmatrix},
\end{equation}
we have
\begin{equation}
\langle u|Y=(0\ 0\ 1)
\begin{pmatrix}
0 & 1 & 0\\
-1 & 0 & 0\\
0 & 0 & 0
\end{pmatrix}
=(0\ 0\ 0).
\end{equation}
Hence along any \(Y\)-arc,
\begin{equation}
\frac{d}{dt}\langle u|\psi(t)\rangle
=\langle u|\dot\psi(t)\rangle
=\langle u|Y|\psi(t)\rangle
=0.
\end{equation}
Therefore \(\langle u|\psi(t)\rangle\) is constant on every \(Y\)-arc. A final \(Y\)-arc cannot change \(\langle u|\psi\rangle\), and so cannot steer a point outside \(\Sigma\) into \(\Sigma\). Thus any nontrivial time-optimal trajectory must terminate with an \(X\)-arc.
\end{proof}

\begin{lemma}\label{lem:first-arc-global-app}
Every regular time-optimal trajectory begins with an \(X\)-arc.
\end{lemma}

\begin{proof}
In the large-block regime, the initial state \(|s_n\rangle\) in the reduced basis \((|t\rangle,|ntt\rangle,|u\rangle)\) has the leading-order form
\begin{equation}
|s_n\rangle \sim
\begin{pmatrix}
0\\
\sin\gamma\\
\cos\gamma
\end{pmatrix}.
\end{equation}
For a short time \(\tau>0\), compare the trajectories starting from \(x_0=|s_n\rangle\) under pure \(X\)- and pure \(Y\)-evolution:
\begin{equation}
\psi_X(\tau)=e^{\tau X}x_0,\qquad \psi_Y(\tau)=e^{\tau Y}x_0.
\end{equation}

Since \(Y\) rotates only in the \((|t\rangle,|ntt\rangle)\)-plane and leaves the \(|u\rangle\)-component unchanged,
\begin{equation}
\psi_Y(\tau)=
\begin{pmatrix}
\sin\gamma\,\sin\tau\\
\sin\gamma\,\cos\tau\\
\cos\gamma
\end{pmatrix},
\qquad
\langle u|\psi_Y(\tau)\rangle=\cos\gamma.
\end{equation}
Let \(s=\sin\gamma\) and \(c=\cos\gamma\). From Eq.~\eqref{eq:XY-generators-sec2}, the \(X\)-arc equations are
\begin{equation}
\dot\psi_1=s^2\psi_2+sc\,\psi_3,\qquad
\dot\psi_2=-s^2\psi_1,\qquad
\dot\psi_3=-sc\,\psi_1.
\end{equation}
With \(\psi(0)=x_0=(0,s,c)^T\), we get
\begin{equation}
\dot\psi_1(0)=s^2\cdot s+sc\cdot c=s,
\qquad
\psi_1(0)=0.
\end{equation}
Hence \(\psi_1\) satisfies
\begin{equation}
\ddot\psi_1+s^2\psi_1=0,\qquad \psi_1(0)=0,\qquad \dot\psi_1(0)=s,
\end{equation}
so \(\psi_1(\tau)=\sin(s\tau)\). Integrating the remaining equations yields
\begin{equation}
\psi_X(\tau)=
\begin{pmatrix}
\sin(s\tau)\\
s\cos(s\tau)\\
c\cos(s\tau)
\end{pmatrix},
\qquad
\langle u|\psi_X(\tau)\rangle=c\cos(s\tau)<c
\end{equation}
for every \(\tau>0\). Therefore
\begin{equation}
\langle u|\psi_X(\tau)\rangle < \langle u|\psi_Y(\tau)\rangle.
\end{equation}
So, for the same short duration, an \(X\)-prefix moves the state strictly closer to the terminal plane \(\Sigma\), whereas a \(Y\)-prefix does not.

Assume that a regular time-optimal trajectory starts with a \(Y\)-arc of positive length. Replacing a sufficiently short initial \(Y\)-prefix by an \(X\)-prefix of the same duration strictly decreases the \(|u\rangle\)-component. By smooth dependence of the endpoint map on the switching times, and by regularity of the extremal, this improvement can be compensated by an \(O(\tau^2)\) perturbation of the later switching times while preserving the terminal condition \(\psi(T)\in\Sigma\). The resulting nearby trajectory reaches \(\Sigma\) in strictly smaller total time, contradicting optimality. Hence the first arc must be \(X\).
\end{proof}

We analyze the endpoint arc structure of the control problem, which corresponds to the first and last operators in the optimal partial-search sequence. Lemmas~\ref{lem:last-arc-global-app} and~\ref{lem:first-arc-global-app} are consistent with the discrete optimization of Ref.~\cite{KorepinLiao2006QuestFastPartialSearch}, which shows that the first and last search operators are always global.


\begin{thebibliography}{10}
\expandafter\ifx\csname url\endcsname\relax
  \def\url#1{{\tt #1}}\fi
\expandafter\ifx\csname urlprefix\endcsname\relax\def\urlprefix{URL }\fi
\providecommand{\eprint}[2][]{\url{#2}}

\bibitem{grover1996fast}
Grover L~K 1996 A fast quantum mechanical algorithm for database search {\em Proceedings of the 28th ACM Symposium on Theory of Computing (STOC)\/} pp 212--219

\bibitem{Grover1997needle}
Grover L~K 1997 {\em Physical Review Letters\/} {\bf 79} 325

\bibitem{boyer1996tight}
Boyer M, Brassard G, H{\o}yer P and Tapp A 1998 {\em Fortschritte der Physik\/} {\bf 46} 493--505

\bibitem{zalka1999grover}
Zalka C 1999 {\em Physical Review A\/} {\bf 60} 2746--2751

\bibitem{Grassl2016ApplyingGroverAES}
Grassl M, Langenberg B, Roetteler M and Steinwandt R 2016 Applying grover's algorithm to aes: Quantum resource estimates {\em Post-Quantum Cryptography\/} ({\em Lecture Notes in Computer Science\/} vol 9606) ed Takagi T (Cham: Springer) pp 29--43

\bibitem{Jaques2020GroverAESLowMC}
Jaques S, Naehrig M, Roetteler M and Virdia F 2020 Implementing grover oracles for quantum key search on aes and lowmc {\em Advances in Cryptology -- EUROCRYPT 2020\/} ({\em Lecture Notes in Computer Science\/} vol 12106) ed Canteaut A and Ishai Y (Springer) pp 280--310

\bibitem{Jang2025QuantumAnalysisAES}
Jang K, Baksi A, Kim H, Song G, Seo H and Chattopadhyay A 2025 {\em IACR Communications in Cryptology\/} {\bf 2} 1--57

\bibitem{Chuang1998}
Chuang I~L, Gershenfeld N and Kubinec M 1998 {\em Physical Review Letters\/} {\bf 80} 3408--3411

\bibitem{zhang2021implementation}
Zhang K, Rao P, Yu K, Lim H and Korepin V 2021 {\em Quantum Information Processing\/} {\bf 20} 233

\bibitem{pokharel2024better}
Pokharel B and Lidar D~A 2024 {\em npj Quantum Information\/} {\bf 10} 23

\bibitem{GroverRadhakrishnan2005PartialSearch}
Grover L~K and Radhakrishnan J 2005 Is partial quantum search of a database any easier? {\em Proceedings of the 17th Annual ACM Symposium on Parallelism in Algorithms and Architectures\/} SPAA '05 (New York, NY, USA: Association for Computing Machinery) pp 186--194

\bibitem{Korepin2005OptimizationPartialSearch}
Korepin V~E 2005 {\em Journal of Physics A: Mathematical and General\/} {\bf 38} L731--L738

\bibitem{KorepinGrover2006SimplePartialSearch}
Korepin V~E and Grover L~K 2006 {\em Quantum Information Processing\/} {\bf 5} 5--10

\bibitem{KorepinVallilo2006GroupTheoreticalPartialSearch}
Korepin V~E and Vallilo B~C 2006 {\em Progress of Theoretical Physics\/} {\bf 116} 783--793

\bibitem{Choi2007QuantumPS}
Choi B~S and Korepin V~E 2007 {\em Quantum Information Processing\/} {\bf 6} 243--254

\bibitem{zhong2009quantum}
Zhong P~C, Bao W~S and Wei Y 2009 {\em Chinese Physics Letters\/} {\bf 26} 020301

\bibitem{Zhang2017QuantumPS}
Zhang K and Korepin V~E 2017 {\em Quantum Information Processing\/} {\bf 17} 143

\bibitem{ChoiWalkerBraunstein2007SureSuccessPartialSearch}
Choi B~S, Walker T~A and Braunstein S~L 2007 {\em Quantum Information Processing\/} {\bf 6} 1--8

\bibitem{YeWang2025DeterministicPartialSearchOneSixteenth}
Ye S~M and Wang Y~L 2025 {\em Physics Letters A\/} {\bf 537} 130325

\bibitem{KorepinLiao2006QuestFastPartialSearch}
Korepin V~E and Liao J 2006 {\em Quantum Information Processing\/} {\bf 5} 209--226

\bibitem{JiangWangZhangKorepin2026ExactBoundsPartialSearch}
Jiang Y~B, Wang X~H, Zhang K and Korepin V 2026 {\em Physical Review A\/} {\bf 114} 012412

\bibitem{Pontryagin1962}
Pontryagin L~S, Boltyanskii V~G, Gamkrelidze R~V and Mishchenko E~F 1962 {\em The Mathematical Theory of Optimal Processes\/} (New York: Interscience Publishers)

\bibitem{Liberzon2012}
Liberzon D 2012 {\em Calculus of Variations and Optimal Control Theory: A Concise Introduction\/} (Princeton, NJ: Princeton University Press)

\bibitem{AgrachevSachkov2004}
Agrachev A~A and Sachkov Y~L 2004 {\em Control Theory from the Geometric Viewpoint\/} ({\em Encyclopaedia of Mathematical Sciences\/} vol~87) (Berlin, Heidelberg: Springer)

\bibitem{DongPetersen2010}
Dong D and Petersen I~R 2010 {\em IET Control Theory \& Applications\/} {\bf 4} 2651--2671

\bibitem{KhanejaGlaserBrockett2001}
Khaneja N, Brockett R and Glaser S~J 2001 {\em Physical Review A\/} {\bf 63} 032308

\bibitem{Glaser2015TrainingSchrodinger}
Glaser S~J, Boscain U, Calarco T, Koch C~P, K{\"o}ckenberger W, Kosloff R, Kuprov I, Luy B, Schirmer S, Schulte-Herbr{\"u}ggen T, Sugny D and Wilhelm F~K 2015 {\em The European Physical Journal D\/} {\bf 69} 279

\bibitem{Stefanatos2020RolandCerf}
Stefanatos D and Paspalakis E 2020 {\em Journal of Physics A: Mathematical and Theoretical\/} {\bf 53} 115304

\bibitem{BoscainMasonSigalotti2021}
Boscain U, Mason P and Sigalotti M 2021 {\em PRX Quantum\/} {\bf 2} 030203

\bibitem{CoddingtonLevinson1955}
Coddington E~A and Levinson N 1955 {\em Theory of Ordinary Differential Equations\/} (New York: McGraw-Hill)

\bibitem{Hartman2002}
Hartman P 2002 {\em Ordinary Differential Equations\/} (Philadelphia: SIAM) ISBN 9780898715101

\bibitem{Preskill2018quantumcomputingin}
Preskill J 2018 {\em Quantum\/} {\bf 2} 79

\bibitem{Zhang2020DepthOptimizationQuantumSearch}
Zhang K and Korepin V 2020 {\em Physical Review A\/} {\bf 101} 032346

\bibitem{Brianski2021IntroducingStructure}
Bria{\'n}ski M, Gwinner J, Hlembotskyi V, Jarnicki W, Pli{\'s} S and Szady A 2021 {\em Physical Review A\/} {\bf 103} 062425

\bibitem{Campos2024DepthScalingQAOA}
Campos E, Rabinovich D and Uvarov A 2024 {\em Physical Review A\/} {\bf 110} 012428

\end{thebibliography}

\providecommand{\noopsort}[1]{}\providecommand{\singleletter}[1]{#1}%
\providecommand{\newblock}{}

\end{document}